%% file: main.tex
\documentclass[a4paper,UKenglish,cleveref,Cref]{lipics-v2021}

\nolinenumbers

\title{Verification of Linear Dynamical Systems via O-Minimality of the Real Numbers} 

\author {Toghrul Karimov} 
{Max Planck Institute for Software Systems, Saarland Informatics Campus, Saarbr\"ucken, Germany}
{toghs@mpi-sws.org}
{https://orcid.org/0000-0002-9405-2332}
{}

\authorrunning{T. Karimov} 

\Copyright{Toghrul Karimov} 



\category{} 

\relatedversion{} 

\funding{The author was supported by the DFG grant 389792660 as
	part of the TRR 248 (see \url{https://perspicuous-computing.science})}

\acknowledgements{The author thanks the anonymous reviewers, as well as Fugen Hagihara of Kyoto University for their many helpful comments and corrections.}

\EventEditors{Keren Censor-Hillel, Fabrizio Grandoni, Joël Ouaknine, and Gabriele Puppis}
\EventNoEds{4}
\EventLongTitle{52nd International Colloquium on Automata, Languages, and Programming (ICALP 2025)}
\EventShortTitle{ICALP 2025}
\EventAcronym{ICALP}
\EventYear{2025}
\EventDate{July 8--11, 2025}
\EventLocation{Aarhus, Denmark}
\EventLogo{}
\SeriesVolume{334}
\ArticleNo{173}

\input{macros.tex}

\usepackage[dvipsnames]{xcolor}
\usepackage{bm,bbm,mathtools,tikz}

\begin{document}

\maketitle

\begin{abstract}
	A discrete-time linear dynamical system (LDS) is given by an update matrix $M \in \rel^{d\times d}$, and has the trajectories $\langle s, Ms, M^2s, \ldots \rangle$ for $s \in \rel^d$.
	Reachability-type decision problems of linear dynamical systems, most notably the Skolem Problem, lie at the forefront of decidability: typically, sound and complete algorithms are known only in low dimensions, and these rely on sophisticated tools from number theory and Diophantine approximation.
	Recently, however, o-minimality has emerged as a counterpoint to these number-theoretic tools that allows us to decide certain modifications of the classical problems of LDS without any dimension restrictions.
	In this paper, we first introduce the Decomposition Method, a framework that captures all applications of o-minimality to decision problems of LDS that are currently known to us.
	We then use the Decomposition Method to show decidability of the Robust Safety Problem (restricted to bounded initial sets) in arbitrary dimension:
	given a matrix~$M$, a bounded semialgebraic set~$S$ of initial points, and a semialgebraic set~$T$ of unsafe points, it is decidable whether there exists $\varepsilon > 0$ such that all orbits that begin in the $\varepsilon$-ball around $S$ avoid $T$.
\end{abstract}

\section{Introduction}

Linear dynamical systems (LDS) are mathematical models widely used in engineering and sciences to describe systems that evolve over time.
A \emph{discrete-time} LDS is given by an update matrix $M$.
The \emph{orbit} of a point $s$ under $M$ is the infinite sequence of vectors $(M^n s)_{n\in\nat}$.
In formal verification, the most prominent decision problem of linear dynamical systems is the (point-to-set) Reachability Problem: given $M \in\rat^{d\times d}$, $s \in \rat^d$, and a \emph{semialgebraic} target set~$T$, decide whether there exists $n \in \nat$ such that $M^ns \in T$.
In terms of program verification, this is equivalent to the \emph{Termination Problem} for linear loops: given a program fragment of the form
\begin{align*}
	&x_1,\ldots,x_d \coloneqq s_1,\ldots,s_d\\
	&\mathbf{while} \:\: \lnot \mathrm{P}(x) \:\: \mathbf{do} \:\: x \coloneqq M \cdot x
\end{align*}
where $x_1,\ldots,x_d$ are variable names, $s_i$ is the initial value of $x_i$, $x = (x_1,\ldots,x_d)$, $M$ is a linear update and $P$ is a Boolean combination of polynomial inequalities, decide whether the loop terminates.
The restriction of the Termination Problem to $P$ defined by a single linear equality (i.e., the Reachability Problem where $T$ is restricted to hyperplanes) is Turing-equivalent to the famously open \emph{Skolem Problem} of linear recurrence sequences (LRS): given an LRS $(u_n)_{n\in\nat}$ defined  by a recurrence relation
\[
u_{n+d} = a_1u_{n+d-1} + \cdots + a_d u_n
\]
as well as the initial values $u_0,\ldots,u_{d-1}$, decide whether there exists $n$ such that $u_n = 0$.
The Skolem Problem is known to be decidable\footnote{Decidability for $d \le 4$ applies to LRS over real algebraic numbers. For LRS over algebraic numbers, decidability is known for $d \le 3$.} only for LRS with $d \le 4$ \cite{mignotte-shorey-tijdeman-skolem}.
Consequently, the Termination Problem for linear loops with a linear guard is currently open for $5$ or more program variables.
Similarly, the Termination Problem with $P$ defined by a single linear inequality (i.e., the Reachability Problem with halfspace targets) is Turing-equivalent to the \emph{Positivity Problem}: given an LRS $(u_n)_{n\in\nat}$, decide whether $u_n \ge 0$.
The Positivity Problem subsumes the Skolem Problem \cite[Chap.~2.5]{karimov-thesis} and is decidable for LRS with $d \le 5$, but is also known to be hard with respect to certain long-standing open problems in Diophantine approximation \cite{joel-pos-low-dim}: a resolution (in either direction) of the decidability of the Positivity Problem would entail major breakthroughs regarding approximabiliity of \emph{Lagrange constants} for a large class of transcendental numbers, which are currently believed to be out of reach.

Decidability of the Reachability Problem, which is formidably difficult in full generality, has also been studied under various geometric restrictions on $T$.
The state of the art in this direction is that the Reachability Problem is decidable in arbitrary dimension for $T$ restricted to the class of \emph{tame} targets \cite{karimovPOPL2}, i.e.\ for $T$ that can be constructed through the usual set operations from semialgebraic sets that either (i) have dimension one (i.e.\ are ``string-like''), or (ii) are contained in a three-dimensional subspace of $\rel^d$.
In particular, the Reachability Problem is decidable for arbitrary $T$ in dimension $d\le 3$.
These decidability results rely on Baker's theorem \cite{baker-rational-sharp-version-1993} as well its $p$-adic analogue \cite{yu-padic-baker}, and are tight in the sense that allowing $T$ to be two-dimensional or contained in a four-dimensional subspace yields a decision problem that is Diophantine-hard similarly to the Positivity Problem \cite[Chap.~8]{karimov-thesis}.

As discussed above, the Reachability Problem is intimately related to number theory, and, mathematically speaking, becomes harder and harder as we increase the dimension~$d$.
Recently, however, a string of results have emerged that show decidability of various reachability-type problems of LDS in arbitrary dimension.
First, Akshay et al. \cite{akshay2024robustness} showed that the robust variants of the Skolem Problem as well as the Positivity Problem are decidable: given a matrix $M$, an initial point $s$, and a hyperplane or a halfspace $T$, we can decide whether there exists $\varepsilon > 0$ such that for all $s'$ in the $\varepsilon$-ball around $s$, the orbit $\seq{M^ns'}$ avoids~$T$. 
The proof is based on classical analyses of LRS.
A second interesting result was given by Kelmendi in \cite{kelmendi2023computing}, who showed that given $M$, an initial point $s$, and a semialgebraic target~$T$, the \emph{frequency} of visits
\[
\mu = \lim_{n\to \infty} \frac 1 n \sum_{k=0}^{n-1} \mathbbm{1}(M^k s \in T)
\]
exists and can be effectively compared against 0.
The proof is variation on the fundamental result~\cite{joel-ult-pos-simple} that \emph{ultimate positivity} is decidable for \emph{simple} (also known as \emph{diagonalisable}) linear recurrence sequences.
Lastly, \cite{d2022pseudo} showed that the Pseudo-Orbit Reachability Problem is decidable for diagonalisable~$M$ and a single starting point $s$: given such $M,s$ and a semialgebraic target~$T$, we can decide whether for every $\varepsilon>0$ there exists $(x_n)_{n\in\nat}$ such that (i) $x_0 = s$, (ii) $\Vert x_{n+1} - Mx_n\Vert < \varepsilon$ for all $n$, and (iii) $x_n \in T$ for some $T$.
Such $(x_n)_{n\in\nat}$ is called an \emph{$\varepsilon$-pseudo-orbit} of~$s$ under $M$.
The proof of decidability uses geometry of pseudo-orbits of diagonalisable $M$ and the classical tools of linear dynamical systems in combination with \emph{o-minimality}.
Briefly, o-minimality of real numbers equipped with arithmetic and exponentiation tell us that every subset of $\rel^d$ definable using first-order logic and the aforementioned operations has finitely many connected components.

The thesis of this paper is that every single variant of the Reachability Problem that is decidable in an arbitrary dimension $d$ can, in fact, be explained using o-minimality (\Cref{sec::other-applications}).
Although a concept originating in model theory, a branch of mathematical logic, o-minimality has recently had spectacular applications to Diophantine geometry, including counting rational points in algebraic varieties as well as the proofs of Mordell-Lang and André-Oort conjectures~\cite{pila2014minimality}.
To apply o-minimality, we introduce the Decomposition Method\footnote{A similar decomposition is used in \cite{almagor2020invariants}.}, a framework in which all decidability results mentioned above (and more) 
can be proven.
As an application of the Decomposition Method, we study the Robust Safety Problem, which is motivated by situations in which a system must remain in a safe set forever, even when an adversarial perturbation is applied to the initial state.
For effectiveness reasons, let us work with real algebraic numbers\footnote{If we move from the reals to the complex numbers, i.e.\ consider the Robust Safety Problem in $\alg^d$, all of our results still apply. In this setting semialgebraic sets are defined by identifying $\com^d$ with $\rel^{2d}$.}, denoted by $\ralg$, which subsume rationals and can be effectively represented and manipulated in computer memory.
For $s \in \rel^d$ and $\varepsilon > 0$, denote by $B(s, \varepsilon)$ the open $\varepsilon$-ball around~$s$.
Further write  $B(S,\varepsilon)$ for $\bigcup_{s \in S} B(s,\varepsilon)$.
The Robust Safety Problem is to decide, given $M \in (\ralg)^{d\times d}$, a semialgebraic set $S$ of initial points, and a semialgebraic set $T$ of unsafe points, whether the there exists $\varepsilon > 0$ such that the sequence $(M^n \cdot B(S,\varepsilon))_{n\in\nat}$ \emph{avoids}~$T$: that is, whether $M^n \cdot B(s,\varepsilon)$ does not intersect~$T$ for all $s \in S$ and $n \in \nat$.
Our main result is the following, which, for bounded~$S$, (i) characterises the largest $\varepsilon > 0$ such that $(M^n \cdot B(S,\varepsilon))_{n\in\nat}$ avoids~$T$, and (ii)~shows decidability of the Robust Safety Problem.

\begin{theorem}
	\label{thm:main}
	Let $M \in (\ralg)^{d\times d}$, $S \subseteq \rel^d$ be non-empty, semialgebraic and bounded, and $T \subseteq \rel^d$ be semialgebraic.
	Further let
	\begin{align*}
		\mu_1 &= \sup\, \{\varepsilon \ge 0 \colon M^n\cdot B(S, \varepsilon) \textrm{ does not intersect $T$ for all $n$}\},\\
		\mu_2 &= \sup\, \{\varepsilon \ge 0 \colon M^n\cdot B(S, \varepsilon) \textrm{ does not intersect $T$ for all sufficiently large $n$}\}.
	\end{align*}
	Then $\mu_1 \in \ralg$, $\mu_2 \in (\ralg) \cup \{\infty\}$ and is effectively computable, $\mu_1$ can be approximated (both from above and below) to arbitrary precision, and it is decidable whether $\mu_1 = 0$.
	Moreover, for every $\varepsilon \in (0, \mu_2) \cap \rat$ we can effectively compute~$N$ such that $(M^n \cdot B(S,\varepsilon))_{n\ge N}$ avoids $T$.
\end{theorem}
Intuitively, $\mu_1, \mu_2$  are the best possible margin of safety and ``asymptotic safety'', respectively.
From \Cref{thm:main} it follows that the Robust Safety Problem is decidable: we simply check whether $\mu_1 = 0$.
For positive instances, i.e.\ when $\mu_1 > 0$, we can compute~$\varepsilon \ge 0$ that is arbitrarily close to the best possible safety margin by approximating $\mu_1$ from below.
Finally, for every positive $\varepsilon \ne \mu_2$ we can decide whether $(M^n \cdot B(S,\varepsilon))_{n\in\nat}$ avoids $T$:
for $\varepsilon > \mu_2$, the answer is immediately negative.
For $\varepsilon \in (0,\mu_2)$, we can compute $N$ and then check whether $M^n \cdot B(S,\varepsilon)$ intersects $T$ for some $n < N$.
For $\varepsilon = \mu_2$, however, we do not know how to decide whether $(M^n \cdot B(S,\varepsilon))_{n\in\nat}$ avoids $T$: in this case we have to contend with the hard Diophantine approximation problems that also arise in the classical Reachability Problem.

\paragraph*{Related work}

At the time of writing, to the best of our knowledge, there are only two published works that apply o-minimality to verification of (discrete-time) linear dynamical systems: \cite{almagor2022minimal} studies o-minimal and semialgebraic \emph{invariants} of LDS (see \Cref{sec::other-applications}), and \cite{d2022pseudo} studies the Pseudo-Orbit Reachability Problem for LDS.
In the unpublished extension \cite{d2022pseudo-arxiv} of \cite{d2022pseudo}, it is described how the decidability of the Pseudo-Orbit Reachability Problem can be adapted to study the Robust Safety Problem for singleton $S$.
The latter problem, however, turns out to be much simpler, and can be solved more generally and directly using the Decomposition Method.

\Cref{thm:main} is a generalisation of the aforementioned results of Akshay et al. \cite{akshay2024robustness} to bounded~$S$ and semialgebraic $T$.
A similar result is \cite{d2021pseudo}, which shows the Pseudo-Orbit Reachability Problem is decidable for hyperplane/halfspace targets, without any restrictions on $M$.
Both of these heavily rely on the fact that $T$ is defined by a single linear (in)equality, and their approach does not generalise to targets defined by non-linear or multiple inequalities.

Some other connections between o-minimality and dynamical (or, more generally, cyber-physical) systems have already been established.
In \cite{lafferriere2000minimal}, Lafferriere et al.\ show that \emph{o-minimal hybrid systems} always admit a finite bisimulation; these do not include linear dynamical systems.
In \cite{miller2011expansions}, Miller studies expansions of the field of real numbers with trajectories of LDS from the perspective of definability and o-minimality.

\section{Preliminaries}

\subsection{Notation and conventions}

We denote by $\im$ the imaginary number and by $\torus$ the unit circle in $\com$.
We write $\zerovec$ for a vector of all zeros as well as a zero matrix.
In both cases, the dimensions will be clear from the context.
For $x\in\rel^d$ and $\varepsilon > 0$, we define $B(x,\varepsilon) = \{y\in\rel^d\colon \Vert x-y\Vert < \varepsilon\}$
where $\Vert\cdot\Vert$ is the $\ell_2$-norm.
We work exclusively with open $\ell_2$-balls.
We abbreviate $B(\zerovec,\varepsilon)$ to $B(\varepsilon)$.
For $x\in \rel^d$ and $Y\subseteq \rel^d$, we define $d(x,Y) = \inf \{\Vert y - x \Vert \colon y \in Y\} \in \rel \cup \{\infty\}$.
For a set $Y$ and $\varepsilon > 0$, we let $B(Y, \varepsilon) = \{x \colon d(x,Y) < \varepsilon\}$.

\label{weird-def}
For sets of vectors $X, Y$, we write $X+Y$ for $\{x + y \colon x \in X, y \in Y\}$.
For a matrix $M$ and a set $X$ of vectors, we write $M \cdot X$ to mean $\{Mx\colon x \in X\}$.
Finally, for a set $\Mcal$ of matrices and a set $X$ of vectors we define $\Mcal \cdot X = \{M x \colon M \in \Mcal, x \in X\}$.
For a (not necessarily invertible) matrix $C \in \rel^{d\times d}$ and $X \subseteq \rel^d$, we define  $C^{-1} \cdot X$ to be $\{y \in \rel^d \colon C y \in X\}$.
We write $C^{-n}$ for $(C^n)^{-1}$.

We denote by $\operatorname{Cl}(X)$ the closure of $X$ in a topological space that will be either explicit or clear from the context. 
We will only be working with Euclidean topology and induced subset topologies.
When we say that an object $X$ is effectively computable, we mean that a representation of $X$ in a scheme that will be clear from the context is effectively computable.

\subsection{Linear algebra}

For a matrix $A \in \rel^{d\times d}$, $\Vert A \Vert = \max_{x \ne \zerovec} \Vert A x\Vert / \Vert x \Vert$.
The matrix norm is sub-multiplicative: for all $A,B$ of matching dimensions, $\Vert A B \Vert \le \Vert A \Vert \cdot \Vert B \Vert$.

Let $X_i \in \rel^{d_i \times d_i}$ for $1\le i \le k$.
We write $\diag(X_1,\ldots,X_k)$ for the block-diagonal matrix in $\rel^{d\times d}$, where $d = d_1 + \cdots + d_k$, constructed from $X_1,\ldots,X_k$ respecting the order.
For $a,b \in \rel$, let $\Lambda(a,b) =  \begin{bmatrix}
	a & -b\\
	b & a
\end{bmatrix}$, and for $\theta \in \rel$, let $R(\theta) = \Lambda(\sin(\theta), \cos(\theta))$.
A \emph{real Jordan block} is a matrix 
\begin{equation}
	\label{eq:prelims-jordan-block}
	J = 
	\begin{bmatrix}
		\Lambda & I \\
		&\Lambda & \ddots\\
		&&\ddots & I\\
		&&&\Lambda
	\end{bmatrix} \in \rel^{d \times d}
\end{equation}
where $I$ is an identity matrix, and either (i) $\Lambda, I \in \rel^{1\times 1}$, or (ii) $\Lambda = \rho R(\theta)$ with $I \in \rel^{2\times 2}$, $\rho \in \rel_{> 0}$ and $\theta \in \rel$.
A matrix $J$ is in \emph{real Jordan form} if $J = \diag(J_1,\ldots,J_l)$ where each~$J_i$ is a real Jordan block. 
Given $M \in (\ralg)^{d\times d}$, we can compute $P, J \in (\ralg)^{d\times d}$ such that $J$ is in real Jordan form and $M = P^{-1}JP$ \cite{cai1994computing}.

\subsection{Logical theories}

A \emph{structure} $\mathbb{M}$ consists of a universe $U$, constants $c_1,\ldots,c_k \in U$, predicates $P_1,\ldots,P_l$ where each $P_i \subseteq U^{\mu(i)}$ for some $\mu(i) \ge 1$, and functions $f_1,\ldots,f_m$ where each $f_i$ has the type $f_i \colon U^{\delta(i)} \to U$ for some $\delta(i) \ge 1$.
By the \emph{language} of the structure $\mathbb{M}$, written  $\Lcal_{\Mb}$, we mean the set of all well-formed first-order formulas constructed from symbols denoting  the constants $c_1,\ldots, c_k$, predicates $P_1,\ldots,P_l$, functions $f_1,\ldots, f_m$, as well as the equality symbol $=$, the quantifier symbols $\forall, \exists$ and the connectives $\land, \lor, \lnot$.
A \emph{theory} is simply a set of sentences, i.e.\ formulas without free variables.
The theory of the structure $\mathbb{M}$, written $\operatorname{Th}(\Mb)$, is the set of all sentences in the language of $\mathbb{M}$ that are true in $\mathbb{M}$.
A theory $\Tcal$
\begin{itemize}
	\item is \emph{decidable} if there exists an algorithm that takes a sentence $\varphi$ and decides whether $\varphi \in \Tcal$, and
	\item admits \emph{quantifier elimination} if for every formula $\varphi$ with free variables $x_1,\ldots,x_n$, there exists a quantifier-free formula $\psi$ with the same free variables such that the sentence
	\[
	\forall x_1,\ldots,x_n\colon (\varphi(x_1,\ldots,x_n) \Leftrightarrow \psi(x_1,\ldots,x_n))
	\]
	belongs to $\Tcal$.
\end{itemize}

We will be working with the following structures and their theories.
\begin{itemize}
	\item Let $\rel_0  = \langle \rel; 0,1,<,+,\cdot \rangle$, which is the ordered ring of real numbers.
	We will denote the language of this structure by $\Lcal_{or}$, called the \emph{language of ordered rings}.
	Observe that using the constants $0,1$ and the addition, we can obtain any constant $c \in \nat$.
	Hence every atomic formula in $\Lcal_{or}$ with $k$ free variables is equivalent to $p(x_1,\ldots,x_k) \sim 0$, where $p$ is a polynomial with integer coefficients and $\sim$ is either $>$ or the equality.
	By the Tarski-Seidenberg theorem, $\operatorname{Th}(\rel_0)$ admits quantifier elimination and is decidable \cite{bochnak2013real}.
	\item Let $\rexp = \langle \rel; 0,1,<,+,\cdot, \exp \rangle$, the ring of real numbers augmented with the function $x \mapsto e^x$.
	The theory of this structure is \emph{model-complete}, meaning that quantifiers in any formula can be eliminated down to a single block of existential quantifiers, and decidable assuming Schanuel's conjecture \cite{vdD1994bounded-analytic,macintyre1996decidability}.
\end{itemize}

We say that a structure $\Sb$ \emph{expands} $\Mb$ if $\Sb$ and $\Mb$ have the same universe and every constant, function, and relation of $\Mb$ is also present in~$\Sb$.
We will only need structures expanding $\rel_0$.
A set $X \subseteq U^d$ is \emph{definable} in a structure $\mathbb{M}$ if there exist $k \ge 0$, a formula $\varphi$ in the language of $\mathbb{M}$ with~$d+k$ free variables, and $a_1,\ldots,a_k \in U$ such that for all $x_1,\ldots,x_d \in U$, $\varphi(x_1,\ldots,x_d, a_1,\ldots,a_k)$ holds in $\Mb$ if and only if $(x_1,\ldots,x_d) \in X$.
We say that $X$ is \emph{definable in $\Mb$ without parameters} if we can take $k = 0$ above.
Similarly, a function is definable (without parameters) in $\Mb$ if its graph is definable (without parameters) in $\Mb$.

A structure $\Mb$ expanding $\rel_0$ is \emph{o-minimal} if every set definable in $\Mb$ has finitely many connected components.
The structures $\rel_0$ and $\rexp$,  as well as the expansion of $\rexp$ with bounded trigonometric functions \cite{vdD1994bounded-analytic}, are o-minimal.

\subsection{Semialgebraic sets and algebraic numbers}
\label{sec:semialgebraic-sets}
A set $X \subseteq \rel^d$ is \emph{semialgebraic} if it is definable in $\rel_0$ without parameters.
By quantifier elimination, every semialgebraic set can be defined by a Boolean combination of polynomial equalities and inequalities with integer coefficients.
We say that $Z \subseteq \com^d$ is semialgebraic if $\widetilde{Z} = \{(x_1,y_1,\ldots,x_d,y_d) \colon (x_1+\im y_1,\ldots,x_d + \im y_d) \in Z\}$ is a semialgebraic subset of $\rel^{2d}$.
We represent $Z$ by a formula defining $\tilde{Z}$.
A function $f \colon X \to Y$  is semialgebraic if its graph $\{(x, f(x)) \colon x \in X\} \subseteq X \times Y$ is semialgebraic.
We also identify $\rel^{a\times b}$ with $\rel^{ab}$, and define semialgebraic subsets of $\rel^{a\times b}$ accordingly.

A number $z \in \com$ is algebraic if there exists a polynomial $p \in \intg[x]$ such that $p(z) = 0$.
The set of all algebraic numbers is denoted by $\alg$.
Computationally, we will be working with \emph{real algebraic} numbers.
The number $x \in \ralg$ will be represented by a formula $\varphi \in \Lcal_{or}$ defining the singleton set $\{x\}$.
In this representation, arithmetic on real algebraic numbers is straightforward, and first-order properties of a given number can be verified using the decision procedure for $\operatorname{Th}(\rel_0)$.

\section{The Decomposition Method}

We say that a matrix $M$ is a \emph{scaling matrix} if all eigenvalues of $M$ belong to $\rel_{\ge 0}$, and a \emph{rotation matrix} if $M$ is diagonalisable and all eigenvalues of $M$ have modulus 1.
A decomposition of $M\in\rel^{d\times d}$ is a pair of matrices $(C,D)$ such that $C$ is a scaling matrix, $D$ is a rotation matrix, and $M = CD = DC$.

\begin{lemma}
	Every $M \in \rel^{d\times d}$ has a decomposition $(C,D)$.
	If $M \in (\ralg)^{d\times d}$, then a decomposition satisfying $C, D \in (\ralg)^{d\times d}$ can be effectively computed.
\end{lemma}
\begin{proof}
	Write $M = P^{-1}JP$, where $J$ is in real Jordan form.
	If $M \in (\ralg)^{d\times d}$, then we additionally ensure the same for $P$ and $J$.
	Once we construct a decomposition $(C_J, D_J)$ of~$J$, we have the decomposition $(P^{-1}C_JP, P^{-1}D_JP)$ of $M$.
	Suppose $J = \diag(J_1,\ldots,J_m)$, where each $J_i$ is a real Jordan block.
	It suffices to construct decompositions $(C_i,D_i)$ of $J_i$ for $1\le i \le m$: then $(\diag(C_1,\ldots,C_m), \diag(D_1,\ldots,D_m))$ is a decomposition of $J$.
	
	Suppose $J_i$ has a single real eigenvalue $\rho$.
	If $\rho \ge 0$, then we take $D = I$ and $C = J_i$; otherwise, we take $D = -I$ and $C = -J_i$.
	Now suppose $J$ is of the form \eqref{eq:prelims-jordan-block}, where $\Lambda = \rho R(\theta)$  for some $\rho> 0$ and $I$ is the $2\times 2$ identity matrix.
	Note that the entries of $R(\theta)$ as well as $\rho$ are real algebraic in case $M \in (\ralg)^{d\times d}$.
	Let $N$ be the nilpotent matrix satisfying $J_i = \diag(\Lambda,\ldots,\Lambda) + N$.
	We compute the decomposition $(C_i, D_i)$ with $C_i = \rho I + N$ and $D_i = \diag(R(\theta),\ldots,R(\theta))$.
	The only eigenvalue of $C_i$ is $\rho$, and the eigenvalues of $D_i$ are $e^{\im \theta}, e^{-\im\theta} \in \torus \cap \alg$.
	That $C_i$ and $D_i$ commute can be verified directly.
\end{proof}

We will apply the Decomposition Method to the Robust Safety Problem as follows.
Let $M \in \rel^{d\times d}$ and $S, T \subseteq \rel^d$.
We have that for every $\varepsilon > 0$, $M^n \cdot B(S, \varepsilon)$ does not intersect $T$ for all $n \in\nat$ if and only if
\[
D^n \cdot B(S,\varepsilon) \textrm{ does not intersect } C^{-n} \cdot T
\]
for all $n \in \nat$.
We will study the sequences $(D^n \cdot B(S,\varepsilon))_{n\in\nat}$ and $(C^{-n}\cdot T)_{n\in\nat}$ separately.
The sequence $(D^n \cdot B(S,\varepsilon))_{n\in\nat}$ is, in the parlance of dynamical systems theory, ``uniformly recurrent''; this is proven using Kronecker's theorem in Diophantine approximation.
The sequence $(C^{-n}\cdot T)_{n\in\nat}$, on the other hand, converges to a limit shape in a strong sense thanks to o-minimality.
We will then combine our analyses of the two sequences to prove \Cref{thm:main}.

\subsection{Applications of Kronecker's theorem}
\label{sec:kronecker}

We next develop the tools necessary for analysing the sequence $(D^n \cdot B(S, \varepsilon))_{n\in\nat}$ where $D$ is a rotation matrix and $S$ is bounded and semialgebraic.
Our main tool is Kronecker's classical theorem in simultaneous Diophantine approximation.
For $x, y \in \rel$ define $[\![ x ]\!]_y$ to be the closest distance from $x$ to an integer multiple of $y$, and write $[\![ x ]\!] = [\![ x ]\!]_1$.

\begin{theorem}[Kronecker, see \unexpanded{\cite[Chap.~III.5]{cassels-book}}]
	\label{thm:kronecker-classical}
	Let $\lambda_1,\ldots,\lambda_l$ and $x_1,\ldots,x_l$ be real numbers such that for all $u_1,\ldots,u_l \in \intg$,
	\[
	u_1 \lambda_1 + \cdots + u_l \lambda_l \in \intg \Rightarrow u_1x_1 + \cdots + u_lx_l \in \intg.
	\]
	Then for every $\varepsilon > 0$ there exist infinitely many $n \in \nat$ such that $[\![ n \lambda_i - x_i ]\!] < \varepsilon$ for all $i$.
\end{theorem}

Let $\beta_1,\ldots,\beta_l \in \torus \cap \alg$, which in our case will be the eigenvalues of a rotation matrix~$D$.
Write $\beta = (\beta_1,\ldots,\beta_l) \in \torus^l$ and $\beta^k = (\beta_1^k,\ldots,\beta_l^k)$ for $k \in \intg$.
We will use Kronecker's theorem to analyse the orbit $\seq{\beta^n}$ in $\torus^l$, from which we will deduce various properties of the sequence $(D^n)_{n\in\nat}$ in $\rel^{d\times d}$.
Let
\[
G(\beta) = \{(k_1,\ldots,k_l) \in \intg^l \colon \beta_1^{k_1}\cdots \beta_l^{k_l} = 1\}
\]
which is called the \emph{group of multiplicative relations} of $(\beta_1,\ldots,\beta_l)$.
Observe that $G(\beta)$ is an abelian subgroup of $\intg^l$ and thus has a finite basis of $m \le l$ elements.
Such a basis can be computed using the result \cite{masse-mult-rel-bound} of Masser, which places an effective upper bound $M(\beta)$ on the smallest bit length of a basis of $G(\beta)$ in terms of the description length of $\beta$.
To compute a basis, it remains to enumerate all linearly independent subsets of $\intg^l$ of bit length at most $M(\beta)$, and select a maximal one that only contains multiplicative relations satisfied by $\beta$.

Let $X(\beta) = \{z \in \torus^l \colon G(\beta) \subseteq G(z)\} \subseteq \torus^l$. 
We next prove that $(\beta^n)_{n\in\nat}$ is uniformly recurrent in $X(\beta)$.

\begin{lemma}
	\label{thm:kronecker-uniform-recurrence}
	The set $X(\beta)$ is compact and semialgebraic with an effectively computable representation.
	Moreover, for every open subset $O$ of $\torus^l$ containing some $\alpha \in X(\beta)$ there exist infinitely many $n \in \nat$ such that $\beta^n \in O$.
\end{lemma}
\begin{proof}
	To compute a representation of $X(\beta)$, first compute a basis $V = \{v_1,\ldots,v_m\}$ of $G(\beta)$ as described above.
	Write $v_i = (v_{i,1},\ldots,v_{i,l})$ for $1 \le i \le m$.
	Then $X(\beta) = \bigcap_{i=1}^l X_i$ where 
	\[
	X_i = 
	\{
	(z_1,\ldots,z_l) \in \torus^l \colon z_1^{v_{i,1}} \cdots z_l^{v_{i,l}} = 1
	\}.
	\]
	It remains to observe that each $X_i$ is closed and semialgebraic.
	
	To prove the second claim, let $\alpha = (e^{\bm{i} a_1},\ldots,e^{\bm{i} a_l}) \in X(\beta)$, where $a_i\in \rel$ for all $i$.
	Write $\beta_i = e^{\im b_i}$ where $b_i \in \rel$ for all $i$.
	For all $n\in \nat$,
	\[
	\Vert \beta^n - \alpha \Vert_\infty = \max_{1 \le i \le l} |e^{\bm{i} n b_i} - e^{\bm{i} a_i}| 
	\le 
	\max_{1 \le i \le l}\, [\![ nb_i - a_i ]\!]_{2\pi} 
	= 2\pi 	\cdot \max_{1 \le i \le l}\,  [\![ nb_i/(2\pi) - a_i/(2\pi)]\!].
	\]
	We will apply Kronecker's theorem to the right-hand side of the last equality and to show that it can be made arbitrarily small for infinitely many $n$.
	To do this, first we need to show that  for all $u_1,\ldots, u_l \in \intg$, 
	\[
	\frac{u_1b_1}{2\pi} + \cdots + \frac{u_lb_l}{2\pi} \in \intg \Rightarrow \frac{u_1a_1}{2\pi} + \cdots + \frac{u_la_l}{2\pi} \in \intg.
	\] 
	Suppose $\frac{u_1b_1}{2\pi} + \cdots + \frac{u_lb_l}{2\pi} \in \intg$, and observe that this is equivalent to $e^{\bm{i} u_1 b_1} \cdots e^{\bm{i} u_l b_l} = 1$.
	Since $\alpha \in X(\beta)$, $G(\beta) \subseteq G(\alpha)$ and hence $e^{\bm{i} u_1 a_1} \cdots e^{\bm{i}  u_l a_l} = 1$.
	The latter, in turn, is equivalent to $\frac{u_1a_1}{2\pi} + \cdots + \frac{u_la_l}{2\pi} \in \intg$.
	Applying \Cref{thm:kronecker-classical} we obtain that for every $\varepsilon > 0$ there exist infinitely many $n \in \nat$ such that $[\![ nb_i/(2\pi) - a_i/(2\pi)]\!] < \varepsilon$.
	Hence for every $\varepsilon > 0$ there exist infinitely many $n \in \nat$ such that $\Vert \beta^n - \alpha \Vert_\infty < \varepsilon$.
	It remains to construct, for the given open set $O$, a value $\varepsilon > 0$ such that $\{z \in \torus^l \colon \Vert z - \alpha \Vert_\infty < \varepsilon\} \subseteq O$.
\end{proof}
\begin{corollary}
	\label{thm:kronecker-closure}
	The set $X(\beta)$ is the topological closure\footnote{This holds both in $\torus^l$ and $\com^l$, since the former is itself closed in the latter.} of 
	$\{\beta^n \colon n\in\nat\}$.
\end{corollary}
\begin{proof}
	For every $n\in\nat$, $G(\beta) \subseteq G(\beta^n)$ and hence $\beta^n \in X(\beta)$.
	From the density of $(\beta^n)_{n\in\nat}$ in $X(\beta)$ (\Cref{thm:kronecker-uniform-recurrence}) it follows that $X(\beta)$ is the closure of $(\beta^n)_{n\in\nat}$ in $\torus^l$.
\end{proof}

\subsection{Dynamics of rotation matrices}
\label{sec:rotation}

In this section, let $D \in (\ralg)^{d\times d}$ be a rotation matrix and $\Dcal \subseteq \rel^{d\times d}$ be the closure of $\{D^n \colon n\in\nat\}$ in $\rel^{d\times d}$.

\begin{lemma}
	\label{thm:rotation-matrix-closure}
	We have the following.
	\begin{itemize}
		\item[(a)] The set $\Dcal$ is compact, semialgebraic, and effectively computable.
		\item[(b)] For every $A \in \Dcal$, $\det(A) = 1$.
		Moreover, there exists $b >0$ such that $\Vert A \Vert, \Vert A^{-1} \Vert \le b$ for all $A \in \Dcal$.
		\item[(c)] For every non-empty open subset $O$ of $\Dcal$ there exist infinitely many $n \in\nat$ such that $D^n \in O$.
	\end{itemize}
\end{lemma}
\begin{proof}
	Let $\beta_1,\ldots,\beta_l \in \torus \cap \alg$ be the eigenvalues of $D$. Write $\beta_i = e^{\im b_i}$, where $b_i \in\rel$, and let $\beta^n = (\beta_1^n,\ldots,\beta_l^n)$ for $n \in \nat$.
	Define $f \colon \torus^l \to \rel^{d\times d}$ by
	\[
	f(e^{\im \theta_1},\ldots,e^{\im \theta_l}) = \diag(R(\theta_1),\ldots,R(\theta_l), I),
	\]
	where $I$ is the identity matrix of the suitable dimension, and let $U = f(\torus^l)$.
	We have that $f$ is a continuous bijection (i.e.\ a homeomorphism) between $\torus^l$ and~$U$.
	Next, write $D = P^{-1} J P$, where $P \in (\ralg)^{d\times d}$ and $J = f(\beta_1,\ldots,\beta_l) \in (\ralg)^{d\times d}$.
	Observe that for all $n \in \nat$, $J^n = f(\beta^n)$.
	Denote by $\Jcal$ the closure of $\{J^n \colon n \in \nat\}$.
	Since $f$ is a homeomorphism, $\Jcal = f(X(\beta))$, where $X(\beta)$ is the closure of $\{\beta^n \colon n \in \nat\}$ (\Cref{thm:kronecker-closure}).
	Because $f$ is a semialgebraic function (\Cref{sec:semialgebraic-sets}) and $X(\beta)$ is semialgebraic and effectively computable (\Cref{thm:kronecker-uniform-recurrence}), we conclude the same for $\Jcal$.
	To prove~(a) it remains to observe that $\Dcal = P^{-1} \Jcal P$ and hence $\Dcal$ is semialgebraic and effectively computable.
	Since $X(\beta)$ is compact, $\Jcal$ and hence $\Dcal$ are also compact.
	
	Due to the structure of $J$, for every $n \in \nat$ and $x \in \rel^d$, $\det(J^n) = 1$ and 
	\[
	\bigl\Vert J^n x \bigr\Vert = \bigl\Vert J^{-n} x \bigr\Vert = \Vert x \Vert
	\] which implies that $\bigl\Vert J^n \bigr\Vert = \bigl\Vert J^{-n} \bigr\Vert = 1$.
	Choose $b = \Vert P^{-1} \Vert \cdot \Vert P \Vert$.
	We have that for every $n \in \nat$,
	$\det(D^n) = \det(P^{-1})  \det(J^n) \det(P) = 1$,  
	$\bigl\Vert  D^n \bigr\Vert   \le \Vert P^{-1} \Vert\cdot  \bigl\Vert  J^n \bigr\Vert \cdot  \Vert P \Vert \le b$, and similarly $\bigl\Vert  D^{-n} \bigr\Vert \le b$.
	Statement~(b) then follows from the continuity of the determinant, the matrix norm, and the matrix inverse.
	Finally, by \Cref{thm:kronecker-uniform-recurrence}, for every non-empty open subset~$O$ of $X(\beta)$ there exist infinitely many $n \in \nat$ such that $\beta^n \in O$.\@
	Statement~(c) then follows from the fact that $f$ is a homeomorphism.
\end{proof}


We can now give the main result of this section, which is about the orbit of an open set under a rotation matrix.

\begin{lemma}
	\label{thm:rotation-with-open sets}
	Let $O \subseteq \rel^d$ be open.
	\begin{itemize}
		\item[(a)] For all $n \in \nat$, $D^n \cdot O \subseteq \Dcal \cdot O$.
		\item[(b)] For every $y \in \Dcal \cdot O$ there exists $\varepsilon>0$ such that $D^n \cdot O \supset B(y, \varepsilon)$ for infinitely many $n \in \nat$.
	\end{itemize}
\end{lemma}
\begin{proof}
	Statement~(a) follows immediately from the fact that $D^n \in \Dcal$ for all $n \in \nat$.
	To prove~(b), consider $y \in \Dcal \cdot O$.
	By definition, there exist $A \in \Dcal$ and $z \in O$ such that $y = Az$.
	Let $\varepsilon_1,\varepsilon_2 > 0$ be such that $O \supseteq B(z, \varepsilon_1+\varepsilon_2)$.
	We choose $\varepsilon = \varepsilon_1 / b$, where $b$ is such that $\Vert  X \Vert, \Vert X^{-1} \Vert  \le b$ for all $X \in \Dcal$ (\Cref{thm:rotation-matrix-closure}~(b)).
	Recalling that every $X \in \Dcal$ is invertible, let $\delta > 0$ be such that for all $X \in \Dcal$ with $\Vert X - A \Vert < \delta$, we have that $ y\in X \cdot B(z, \varepsilon_2)$.
	
	Consider $D^n \in \Dcal$ such that $\Vert D^n - A \Vert < \delta$.
	By \Cref{thm:rotation-matrix-closure}~(c), there exist infinitely many such $n$.
	We have that
	\[
	D^n \cdot O \supseteq D^n \cdot B(z,\varepsilon_2) + D^n \cdot B(\varepsilon_1)
	\supseteq y + D^n \cdot B(\varepsilon_1).
	\]
	Since $\Vert D^{-n} \Vert \le b$, we have that $D^n \cdot B(\varepsilon_1) \supseteq B(\varepsilon_1/b) = B(\varepsilon)$.
	Therefore, $D^n \cdot O  \supseteq B(y,\varepsilon)$.
\end{proof}

\subsection{Real exponentiation and o-minimality}
\label{sec:omin}

We now develop the tools necessary for studying dynamics of scaling matrices.
The main results of this section are certain (effective) convergence arguments in o-minimal structures.

We say that a sequence $\seq{Z_n}$ of subsets of $\rel^d$ is \emph{definable} in a structure $\Sb$ expanding $\rel_0$ if there exist $k\ge0$, $a_1,\ldots,a_m \in \rel$, and $\varphi \in \Lcal_{\Sb}$ with $d+m+1$ free variables such that for all $n \in \nat$ and $x = (x_1,\ldots,x_d) \in \rel^d$, $x \in Z_n$ if and only if $\varphi(x_1,\ldots,x_d, n, a_1,\ldots,a_m)$ holds in $\Sb$.
The following form of convergence is also known as \emph{Kuratowski convergence}, and captures the kind of set-to-set convergence with which we will work.
\begin{definition}
	\label{def-limit-shape}
	Let $\seq{Z_n}$ be a sequence of subsets of $\rel^d$, and
	\[
	L = \{x \in \rel^d \colon \liminf_{n \to \infty} d(x, Z_n) = 0\}.
	\]
	We say that $\seq{Z_n}$ converges if $L = \{x \in \rel^d \colon \lim_{n \to \infty} d(x, Z_n) = 0\}$, in which case $L$ is called the limit shape of $\seq{Z_n}$.
\end{definition}
We immediately have the following.

\begin{lemma}
	\label{thm:limshape-properties}
	Suppose $L$ is the limit shape of $(Z_n)_{n \in \nat}$. 
	Then
	\begin{itemize}
		\item[(a)] $L$ is closed, 
		\item[(b)] $L = \varnothing$ if and only if for every compact $X$, there exists $N$ such that $Z_n \cap X = \varnothing$ for all $n \ge N$, and
		\item[(c)] for every compact $X \subseteq \rel^d$ and $\varepsilon > 0$, there exists $N$ such that for all $n \ge N$, 
		\[
		Z_n \cap X \subseteq B(L \cap X, \varepsilon).
		\]
		In particular, if $L \cap X = \varnothing$, then $Z_n \cap X = \varnothing$ for all sufficiently large $n$.
	\end{itemize}
\end{lemma}
\begin{proof}
	Statement~(a) is immediate from the definition, and (b) follows from the Bolzano-Weierstra{\ss} theorem.
	To prove~(c), fix compact $X$ and $\varepsilon > 0$.
	The set $Y \coloneqq X \setminus B(L\cap X, \varepsilon)$ is compact.
	For every $y \in Y$, since $y \notin L$, there exist $N_y \in \nat$ and an open subset $B_y \ni y$ of~$Y$ such that $Z_n \cap B_y = \varnothing$ for all $n \ge N_y$.
	Hence $\Ccal = \{B_y \colon y \in Y\}$ is an open cover of $Y$.
	Let~$\Fcal$ be a finite sub-cover.
	We can then take $N = \max \{N_y \colon y \in \Fcal\}$.
	The second part of~(c) follows from the fact that $B(\varnothing, \varepsilon) = \varnothing$.
\end{proof}

We mention that by the properties above, for $\seq{Z_n}$ contained in a compact set $X$ our notion of convergence coincides with convergence with respect to the Hausdorff metric.
Next, we recall that a function definable in an o-minimal structure is ultimately monotonic \cite[Sec.~4.1]{vdD-geometric-categories}. 
The following is an immediate consequence, but we give a little more detail to illustrate the role of o-minimality.

\begin{theorem}
	\label{thm:limshape}
	Every $(Z_n)_{n \in \nat}$ definable in an o-minimal expansion $\Sb$ of~$\rel_0$ converges.
\end{theorem}
\begin{proof}
	Let $a_1,\ldots,a_m \in \rel$ and $\varphi \in \Lcal_{\Sb}$ be a formula with $d + m +1$ free variables such that 
	\[
	(x_1,\ldots,x_d) \in Z_n \Leftrightarrow  \varphi(x_1,\ldots,x_d, n, a_1,\ldots,a_m) \textrm{ holds in } \Sb
	\]
	for all $(x_1,\ldots,x_d)\in\rel^d$ and $n\in\nat$.
	Consider $x \in \rel^k$ with $\liminf_{n \to \infty} d(x, Z_n) = 0$ and $\Delta > 0$.
	Let 
	\[
	T = \{t\ge 0 \mid\exists y = (y_1,\ldots,y_d) \colon d(x, y) < \Delta \:\land\: \varphi(y_1,\ldots,y_d, t, a_1,\ldots,a_m)\}
	\]
	which is definable in $\Sb$.
	By o-minimality, $T$ is a finite union of interval subsets of $\rel_{\ge 0}$.
	Since $\liminf_{n \to \infty} d(x, Z_n) = 0$, $T$ must be unbounded.
	Therefore, $T$ must enclose an interval of the form $[N,\infty)$.
	That is, $d(x, Z_n) < \Delta$ for all sufficiently large $n$.
	It follows that $\lim_{n \to \infty} d(x,Z_n) = 0$.
\end{proof}

In the cases we will encounter, the limit shape $L$ will be a semialgebraic set whose representation can be computed effectively.
To show this, we first need a lemma.

\begin{lemma}
	\label{thm:omin-effective}
	Let $k > 0$, $\varphi \in \Lcal_{or}$ with $k+m+1$ free variables, and $\rho_1,\ldots,\rho_m \in \rel_{>0} \cap  \alg$.
	Define
	\[
	A = \{(x_1,\ldots,x_k) \in \rel^{k} \mid \exists N.\, \forall n \ge N \colon \varphi(x_1,\ldots,x_k, n,  \rho_1^n, \ldots, \rho_m^n)\}
	\] and 
	\[
	B = \{(x_1,\ldots,x_k) \in \rel^{k} \mid \exists N.\, \forall n \ge N \colon \lnot\varphi(x_1,\ldots,x_k, n,  \rho_1^n, \ldots, \rho_m^n)\}.
	\]
	We have the following.
	\begin{enumerate}
		\item[(a)] $A \cup B = \rel^{k}$.
		\item[(b)] Both $A$ and $B$ are semialgebraic sets whose representations can be computed effectively.
		\item[(c)] Given $(x_1,\ldots,x_k )\in B \cap \rat^{k}$, we can effectively compute $N \in \nat$ such that  for all $n \ge N$, $\varphi(x_1,\ldots,x_k, n, \rho_1^n,\ldots,\rho_m^n)$ does not hold.
	\end{enumerate}
\end{lemma}
\begin{proof}
	Let $(x_1,\ldots,x_k) \in \rel^k$, and define $T = \{ t \ge 0 \colon \varphi(x_1,\ldots,x_k, t,  \rho_1^t,\ldots,\rho_m^t)\}$.
	By o-minimality of $\rel_{\exp}$, either $T$ is bounded, or it encloses an interval of the form $[N,\infty)$.
	That is, $(x_1,\ldots,x_k)$ belongs to either $A$ or $B$.
	This proves~(a).
	
	Next, using quantifier elimination in $\operatorname{Th}(\rel_0)$, compute a formula 
	\begin{equation}
		\label{eq:omin-1}
		\bigvee_{i \in I} \bigwedge_{j \in J} p_{i,j}(x_1, \ldots,x_k, y_0,\ldots,y_m) \mathrel{\Delta_{i,j}} 0
	\end{equation}
	in $\Lcal_{or}$ equivalent to $\varphi(x_1,\ldots,x_k,y_0,\ldots,y_m)$, where each $p_{i,j}$ is a polynomial with rational coefficients and $\Delta_{i,j} \in \{>, =\}$.
	Fix $i \in I, j \in J$ and write
	\[
	p_{i,j}(x_1, \ldots,x_k, n,  \rho_1^n,\ldots,\rho_m^n) = \sum_{l=1}^d h_l(x_1,\ldots,x_k) q_l(n) R_l^n
	\]
	where $R_1 > \ldots > R_d > 0$, each $R_i$ real algebraic, and $h_l,q_l$ are non-zero polynomials.
	Without loss of generality we can further assume that $q_l(n) > 0$ for all sufficiently large $n$.
	Let
	\[
	A_{i,j} = \{(x_1,\ldots,x_k) \in \rel^k \colon p_{i,j}(x_1,\ldots,x_k, n, \rho_1^n,\ldots,\rho_m^n) \mathrel{\Delta_{i,j}} 0 \textrm{ for all sufficiently large $n$}\}.
	\]
	Observe that whether $(x_1,\ldots,x_k) \in A_{i,j}$ only depends on the sign of $h_l(x_1,\ldots,x_k)$ for $1 \le l \le d$.
	In particular, if $\Delta_{i,j}$ is the equality, then $A_{i,j}$ is defined by $\bigwedge_{l=1}^d h_l(x_1,\ldots,x_k) = 0$.
	Otherwise, $A_{i,j}$ is defined by 
	\[
	\bigvee_{l=1}^d  \biggl(
	h_l(x_1,\ldots,x_k) > 0 \land \bigwedge_{s=1}^{l-1} h_s(x_1,\ldots,x_k) = 0
	\biggr). 
	\]
	Hence $A_{i,j}$ is semialgebraic with an effectively computable representation. 
	It remains to observe that $A = \bigcup_{i \in I} \bigcap_{j \in J} A_{i,j}$.
	
	To prove~(c), fix $x= (x_1,\ldots,x_k) \in B \cap \rat^k$.
	Let $\varphi \in \Lcal_{or}$ be the formula obtained by substituting the values of $x_1,\ldots,x_k$ into \eqref{eq:omin-1}, which will be of the form
	\[
	\psi(y_0,\ldots,y_m) \coloneqq \bigvee_{i \in I} \bigwedge_{j \in J} v_{i,j}(y_0,\ldots,y_m) \mathrel{\Delta_{i,j}} 0
	\]
	for non-zero polynomials $v_{i,j}$ with rational coefficients and $\Delta_{i,j} \in \{>, =\}$.
	We have to construct $N$ such that for all $n \ge N$, $\psi(n, \rho_1^n,\ldots,\rho_m^n)$ does not hold.
	To do this, it suffices to construct, for each $i \in I$ and $j \in J$, an integer $N_{i,j}$ such that either $v_{i,j}(n, \rho_1^n,\ldots,\rho_m^n) \mathrel{\Delta_{i,j}} 0$ holds for all $n \ge N_{i,j}$, or it does not hold for all $n \ge N_{i,j}$.
	We can then take $N = \max_{i,j} N_{i,j}$.
	
	Fix $i,j$, and write $v_{i,j}(n,\rho_1^n,\ldots,\rho_m^n) = \sum_{l=1}^d q_l(n)R_l^n$ where  $R_1 > \cdots > R_d > 0$, and each $q_l$ is a non-zero polynomial with rational coefficients.
	Compute rationals $M_1, M_2, c > 0$ and $N_{i,j} \in \nat$ such that 
	\begin{itemize}
		\item $R_1 > M_1 > M_2 > R_2$,
		\item  for all $t \ge N_{i,j}$, $|q_1(t)| > c$, and 
		\item $c M_1^n > (d-1)q_l(n)M_2^n$ for all $n \ge N_{i,j}$.
	\end{itemize}
	The sign of $v_{i,j}(n,\rho_1^n,\ldots,\rho_m^n)$ is stable for $n \ge N_{i,j}$ and the same as the sign of $q_1(n)$.
\end{proof}

The following is our main effective convergence result.

\begin{lemma}
	\label{thm:omin-effective-2}
	Let $(Z_n)_{n \in \nat}$ be a sequence of subsets of $\rel^d$.
	Suppose there exist $\varphi \in \Lcal_{or}$ and $\rho_1,\ldots,\rho_m \in \rel_{>0} \cap \alg$ such that for all $x = (x_1,\ldots,x_d) \in \rel^d$ and $n \in \nat$,
	\begin{equation}
		\label{eq:thm:scaling-matrices-lim-shape-1}
		x \in Z_n \Leftrightarrow \varphi(x_1,\ldots,x_d, n, \rho_1^n,\ldots,\rho_m^n).
	\end{equation}
	Then $\seq{Z_n}$ converges, and the limit shape $L$ of $(Z_n)_{\ge 0}$ is semialgebraic and can be effectively computed from $\varphi, \rho_1,\ldots,\rho_m$.
\end{lemma}
\begin{proof}
	The sequence $(Z_n)_{n\in\nat}$ is definable in $\rexp$ and thus converges by \Cref{thm:limshape}.
	Let 
	\[
	X = \{(\varepsilon, x_1,\ldots,x_d) \mid  \varepsilon > 0 \textrm{ and } \exists N.\: \forall n \ge N \colon d((x_1,\ldots,x_d), Z_n) < \varepsilon\}.
	\]
	Invoking \Cref{thm:omin-effective} with $k = d + 1$ we conclude that $X$ is semialgebraic and effectively computable.
	It remains to observe that the limit shape is
	\[
	L = \{(x_1,\ldots,x_d) \mid \forall \varepsilon>0\colon  (\varepsilon, x_1,\ldots,x_d) \in X\}
	\]
	which is semialgebraic and can be effectively computed from $X$.
\end{proof}

\subsection{Dynamics of scaling matrices}
\label{sec:scaling-dynamics}

Now consider a scaling matrix $C \in (\ralg)^{d \times d}$ and a semialgebraic set $T \subseteq \rel^d$.
In order to apply the Decomposition Method, one of the prerequisites is to understand the sequence of sets $(C^{-n} \cdot T)_{n \in \nat}$.
We have the following.
\begin{lemma}
	\label{thm:defining-Zn}
	Suppose the non-zero eigenvalues of $C$ are $\rho_1,\ldots,\rho_m$.
	We can compute $\varphi \in\Lcal_{or}$ such that for all $n\in\nat $ and $(x_1,\ldots,x_d) \in \rel^d$,
	\[
	x \in C^{-n}\cdot T \Leftrightarrow \varphi(x_1,\ldots,x_d, n, \rho_1^{-n}, \ldots,  \rho_m^{-n}).
	\]
\end{lemma}
\begin{proof}
	Note that $\rho_1,\ldots,\rho_m$ must be positive by the assumption that $C$ is a scaling matrix.
	The proof follows immediately from the real Jordan form and the definition of $C^{-n} \cdot T$ (\Cref{weird-def}).
\end{proof}

\begin{lemma}
	\label{thm:scaling-matrices-lim-shape}
	The sequence $(C^{-n} \cdot T)_{n \in \nat}$ converges to a limit shape~$L$ that is semialgebraic and an effectively computable.
\end{lemma}
\begin{proof}
	Let $\rho_1,\ldots,\rho_m$ and $\varphi$ be as above, and apply \Cref{thm:omin-effective-2}.
\end{proof}

\section{Proof of \Cref{thm:main}}

We now prove our main result.
Given $M \in (\ralg)^{d\times d}$, non-empty and bounded semialgebraic~$S$, and semialgebraic $T$, let $(C,D)$ be a decomposition of $M$ with $C,D \in (\ralg)^{d\times d}$.
Recall the definitions of $\mu_1,\mu_2$ from the statement of \Cref{thm:main}, and that for all $n$ and $\varepsilon$, 
\[
M^n \cdot B(S, \varepsilon) \textrm{ intersects $T$} \Leftrightarrow D^n \cdot B(S, \varepsilon) \textrm{ intersects } C^{-n} \cdot T. 
\]
Let $\Dcal$ be the closure of $\{D^n \colon n \in \nat\}$ (\Cref{thm:rotation-matrix-closure}) and $L$ be the limit shape of $(C^{-n}\cdot T)_{n\in\nat}$ (\Cref{thm:scaling-matrices-lim-shape}).
Define $\mu_3 = \sup \, \{\varepsilon \ge 0 \colon \Dcal \cdot B(S,\varepsilon) \textrm{ does not intersect $L$}\}$.

\Cref{fig:1} depicts our application of the Decomposition Method.
The bounded initial set is the triangle $S$.
We are interested in the sequences $(D^n \cdot B(S,\varepsilon))_{n\in\nat}$ for various $\varepsilon \ge 0$: specifically, we want to understand the time steps $n$ at which $D^n \cdot B(S,\varepsilon)$ intersects $C^{-n} \cdot T$.
We see that $\varepsilon = \mu_3$ is the critical value: in this case $\Dcal \cdot B(S,\varepsilon)$, which envelopes $D^n \cdot B(S,\varepsilon)$ for all $n$, just about touches the limit shape $L$.
We have that for $\varepsilon \in (0, \mu_3)$, for all sufficiently large $n$, $C^{-n} \cdot T$ is very close to $L$ and thus is well-separated from $D^n \cdot B(S,\varepsilon)$.
On the other hand, for $\varepsilon > \mu_3$, once again $C^{-n}\cdot T$ stabilises around $L$ for large values of $n$, and by uniform recurrence of $D^n\cdot B(S,\varepsilon)$ in $\Dcal \cdot B(S,\varepsilon)$ (\Cref{thm:rotation-with-open sets}), $D^n \cdot B(S,\varepsilon)$ intersects $C^{-n} \cdot T$ for infinitely many $n$.
We next formalise these arguments.

\begin{figure}
	\centering
	\begin{tikzpicture}[domain=-2.6:2.6]
		\draw[->] (-3.2,0) -- (3.2,0);
		\draw[->] (0,-2.9) -- (0,3.2);
		\fill[fill=Cyan!20!white, even odd rule] (0,0) circle (0.75cm)  (0,0) circle (2.35cm);	
		\fill[fill=Cyan!50!white, even odd rule] (0,0) circle (1.25cm)  (0,0) circle (1.85cm);	
		\filldraw[fill=LimeGreen] (0,1.85) -- (0.3,1.25) -- (-0.3, 1.25) -- cycle;
		\draw[thick,dashed] (0,1.85) -- (0,2.35);
		\node at (0.23, 2.08) {$\mu_3$};
		\draw[ultra thick, Red] (-2.6,2.35) -- (2.6,2.35);
		\draw[WildStrawberry] plot (\x, {0.1 * sin(4*(\x r) + 90) + 2.65});
		\draw[->] (-1.6, 2.71) -- (-1.6,2.41);
		\draw[->] (1.6, 2.71) -- (1.6,2.41);
		\node at (2.75, 2.28) {$L$};
		\node at (1.6, 3) {$C^{-n}\cdot T$};
		\node at (0,1.5) {$S$};
	\end{tikzpicture}
	\caption{Illustration of the Decomposition Method. 
		Here $\Dcal = \{R(\theta) \colon \theta \in \rel\} \subset \rel^{2\times 2}$ is the group of all $2\times 2$ matrices whose action is a rotation, the larger annulus (light blue) is $\Dcal \cdot B(S,\mu_3)$, and the smaller one (blue) is $\Dcal \cdot S$. 
		The sequence ($C^{-n} \cdot T)_{n \in \nat}$ (represented by the curvy line) converges to $L$.
		The length of the dashed line segment in $\mu_3$.}
	\label{fig:1}
\end{figure}

\begin{claim*}
	$\mu_2 = \mu_3$.
\end{claim*}
\begin{claimproof}
	Suppose there exists $\varepsilon \in (0, \mu_3)$.
	We will show that $\varepsilon \le \mu_2$.
	Let $X$ be the closure of $\Dcal \cdot B(S, \varepsilon)$, which is disjoint from $L$.
	Applying \Cref{thm:limshape-properties}~(c), $X \cap (C^{-n} \cdot T)  = \varnothing$ for all sufficiently large $n$, which implies that $D^n \cdot B(S,\varepsilon) \cap  (C^{-n} \cdot T) = \varnothing$. 
	Since the choice of $\varepsilon$ was arbitrary, it follows that $\mu_3 \le \mu_2$.
	
	Now suppose $\varepsilon > 0$ is such that $\varepsilon > \mu_3$.
	Take $x \in L \cap (\Dcal \cdot B(S, \varepsilon))$. 
	Invoking \Cref{thm:rotation-with-open sets} (with $O = B(S,\varepsilon)$), let $\varepsilon' > 0$ be such that $D^n \cdot B(S,\varepsilon) \supseteq B(x, \varepsilon')$ for infinitely many $n \in \nat$.
	Since $x \in L$, by definition of convergence, there exists $N$ such that for all $n \ge N$, $C^{-n} \cdot T$ intersects $B(x, \varepsilon')$.
	Therefore, for infinitely many $n \ge N$, $D^n \cdot B(S,\varepsilon)$ intersects $C^{-n} \cdot T$.
	That is, $\varepsilon \ge \mu_2$.
	It follows that $\mu_2\le \mu_3$.
\end{claimproof}

Observe that the time step $n$ does not appear in the definition of $\mu_3$, which we now know to be equal to $\mu_2$.
Since both $\Dcal$ and $L$ are semialgebraic, we can use a decision procedure for the theory of $\rel_0$ to check whether $\mu_2$ is finite.
If this is the case, we can compute a formula $\varphi \in \Lcal_{or}$ that defines the set $\{\mu_2\}$.
Thus $\mu_2$ must be algebraic. 

Now consider $\varepsilon \in (0, \mu_2) \cap \rat$. 
We have to construct $N$ such that $(M^n \cdot B(S, \varepsilon))_{n\ge N}$ avoids $T$.
Using \Cref{thm:defining-Zn} we can construct a formula $\varphi \in \Lcal_{or}$ such that for all $n$, 
$\varphi(\varepsilon, n, \rho_1^{-n}, \ldots, \rho_m^{-n})$ holds if and only if $\Dcal \cdot B(S,\varepsilon)$ intersects $C^{-n} \cdot T$, where $\rho_1,\ldots,\rho_m$ are the non-zero eigenvalues of $C$.
By the choice of $\varepsilon$, $\varphi(\varepsilon, n, \rho_1^{-n}, \ldots, \rho_m^{-n})$  evaluates to false for all sufficiently large $n$.
Therefore, we can compute the desired $N$ using \Cref{thm:omin-effective}~(c).

We now move on to $\mu_1$.
For $n\in\nat$, let 
\[
\varepsilon_n = \sup\, \{\varepsilon \ge 0 \colon D^n \cdot B(S,\varepsilon) \textrm{ does not intersect }C^{-n} \cdot T\}.
\]
Observe that each $\varepsilon_n$ is algebraic, $\mu_1 = \inf_{n \in \nat}\, \varepsilon_n$, and $\mu_2 = \liminf_{n \in \nat}\, \varepsilon_n$.
Note that each $\varepsilon_n$ is non-negative, but it is possible that $\varepsilon_n > \mu_2$.
We perform a case analysis based on whether $\mu_2$ is infinite.

\paragraph*{Case I.} Suppose $\mu_2$ is infinite, which is the case if and only if $L = \varnothing$.
By \Cref{thm:limshape-properties}~(b), for every compact $X$, $C^{-n} \cdot T$ does not intersect $X$ for all sufficiently large~$n$.
Choose any $\varepsilon > \varepsilon_0$ (which must be positive), and let $X = \operatorname{Cl}(\Dcal \cdot B(S,\varepsilon))$.
Compute $N$ (using \Cref{thm:defining-Zn}, as above) such that $C^{-n} \cdot T$ does not intersect $X$ for all $n \ge  N$.
It follows that $\varepsilon_n \ge \varepsilon$ for all $n \ge N$ and 
\[
\mu_1 = \inf_{n\in\nat} \, \varepsilon_n = \inf \, \{\varepsilon_0, \ldots, \varepsilon_{N-1}\}.
\]
Thus we can explicitly compute $\mu_1$, and it is algebraic.
We can also decide whether $\mu_1 = 0$.

\paragraph*{Case II.} 
Now suppose $\mu_2 \in \ralg$, i.e.\ $L$ is non-empty. 
Then $\mu_1 \in \ralg$ since $\inf_{n\in\nat}\, \varepsilon_n$ is either equal to $\varepsilon_n$ for some~$n$, or to $\liminf_{n\to \infty} \varepsilon_n$.
It remains to show how to decide whether $\mu_1= 0$ and approximate it to arbitrary precision.
Let $\varepsilon \in (0, \mu_2) \in \rat$ and $X$ be the closure of $\Dcal \cdot B(S,\varepsilon)$.
Using \Cref{thm:limshape-properties}~(c), \Cref{thm:defining-Zn}, and \Cref{thm:omin-effective}, compute $N \in \nat$ such that for all $n \ge N$, $C^{-n}\cdot T $ is disjoint from $\Dcal \cdot B(S,\varepsilon)$.
Then $\varepsilon_n \ge \varepsilon$ for all $n \ge N$.
Compute $\xi = \min \{\varepsilon_0, \ldots, \varepsilon_{N-1}\}$.
We have that $\mu_1 = 0$ if and only if $\xi = 0$, which gives us the desired decision procedure.
Our approximation of $\mu_1$ is then $\widetilde{\mu} = \min \{\xi, \varepsilon\}$.
If $\xi \le \varepsilon$, then $\widetilde{\mu} = \mu_1$, i.e.\ we have found the exact value if $\mu_1$.
Otherwise, $\mu_1 \in (\varepsilon, \mu_2)$.
Therefore, we can obtain arbitrarily good sandwiching approximations of $\mu_1$ by taking $\varepsilon \to \mu_2$ from below. 

The reason we are unable to determine $\mu_1$ exactly is as follows.
Suppose we compute $\varepsilon_0, \ldots, \varepsilon_N$ for some large $N$, and observe that $\varepsilon_0 > \cdots > \varepsilon_N > \mu_2$.
There are two possibilities: either $\varepsilon_n$ remains above $\mu_2$ for all~$n$, in which case $\mu_1 = \mu_2$, or $\varepsilon_n < \mu_2$ for some~$n$, in which case $\mu_1$ is an element of the sequence $(\varepsilon_n)_{n\in\nat}$.
At the moment we do not know whether it is possible to determine which of the above is the case by looking at $\varepsilon_0, \ldots, \varepsilon_N$.
 
\section{Other applications of o-minimality}
\label{sec::other-applications}
In this section we briefly discuss how o-minimality and the Decomposition Method are related to various other problems of linear dynamical systems.\footnote{A complete account of the claims in this section, as well as details of how the Pseudo-Orbit Reachability Problem can be solved using the Decomposition Method, will be given in a future paper.}
We first discuss semialgebraic and o-minimal \emph{invariants}.
An invariant of the LDS with update matrix $M \in \rel^{d\times d}$ is a set $\Ical \subseteq \rel^d$ such that $M \cdot \Ical \subseteq \Ical$.
We say that an invariant $\Ical$ of $M$ \emph{separates} an initial set $S$ from a target set $T$ if $S \subseteq \Ical$ and $T \cap \Ical = \varnothing$.
Intuitively, such $\Ical$ certifies that all trajectories starting at $S$ are safe: for all $s \in S$ and $n \in \nat$, $M^ns \notin T$.
In \cite{almagor2022minimal}, the authors prove the following.
Let $M \in (\ralg)^{d\times d}$, $T$ be semialgebraic, and $S = \{s\}$ be a singleton with $s \in (\ralg)^d$.
\begin{itemize}
	\item It is decidable whether $M$ has an invariant $\Ical$ definable in $\rexp$.
	\item If such $\Ical$ exists, then it can be taken to be semialgebraic.
\end{itemize}
Their arguments can be implemented in the framework of the Decomposition Method.
We can moreover easily generalise from singleton to bounded $S$.
\begin{proposition}[\unexpanded{See \cite[Chap.~9.5]{karimov-thesis}}]
	\label{thm:last-sec-invariants}
	Let $S \subseteq \rel^d$ be bounded and semialgebraic, $T$ be semialgebraic, and $M \in (\ralg)^{d\times d}$ with a decomposition $(C,D)$ over $(\ralg)^{d\times d}$.
	Further let $\Dcal$ be the closure of $\{D^n \colon n\in\nat\}$ and $L$ be the limit shape of $\seq{C^{-n} \cdot T}$.
	\begin{enumerate}
		\item For all $n \in \nat$, $M^n \cdot S$ intersects $T$ if and only if $D^n \cdot S$ intersects $C^{-n} \cdot T$.
		\item We can compute $N$ such that either (i) for all $n \ge N$, $\Dcal \cdot S$ intersects $C^{-n} \cdot T$, or (ii) $\Dcal \cdot S$ does not intersect $C^{-n} \cdot T$ for all $n \ge N$.
		\item There exists an $\rexp$-definable invariant $\Ical$ of $M$ separating $S$ from $T$ if and only if (i) holds, in which case $\Ical$ can be taken to be semialgebraic.
	\end{enumerate}
\end{proposition}

Statement~(1) is immediate from the definition of a decomposition, and~(2) is an application of \Cref{thm:omin-effective}: the dichotomy (i-ii), without effectiveness of $N$, can be shown using just o-minimality.
Proving~(3) requires retracing the arguments of \cite{almagor2022minimal}.
Intuitively, \Cref{thm:last-sec-invariants} states that we can topologically separate $M^n \cdot S$ from $T$ using an $\rexp$-definable set if and only if we can topologically separate $\Dcal \cdot S$ (and hence $D^n \cdot S$) from $C^{-n} \cdot T$ for all sufficiently large $n$.
Combining this with our analysis of the Robust Safety Problem yields the following.
\begin{proposition}
	\label{thm:rsp-inv}
	Let $M, S, T$ be as above. 
	If there exists $\varepsilon > 0$ such that $M^n \cdot B(S,\varepsilon)$ does not intersect $T$ for all $n$, then there exists a semialgebraic invariant $\Ical$ separating $S$ from $T$.
\end{proposition}
The idea of the proof is that if $\langle M,S,T \rangle$ is a positive instance of the Robust Safety Problem, then there exists $\varepsilon > 0$ 
  such that for all sufficiently large $n$, $\Dcal \cdot B(S,\varepsilon)$ does not intersect $C^{-n} \cdot T$; this is even stronger than the topological separation described in (ii) of \Cref{thm:last-sec-invariants}~(b).
  The converse of \Cref{thm:rsp-inv} does not hold already in dimension $d = 1$.

We next discuss the main result of \cite{kelmendi2023computing} that, given $M \in (\ralg)^{d\times d}$, an initial point $s \in (\ralg)^d$, and semialgebraic  $T \subseteq \rel^d$, the frequency
\[
\mu = \lim_{n\to \infty} \frac{|\{0\le k < n \colon M^k s\in T\}|}{n}
\]
is well-defined, can be approximated to arbitrary precision, and can be effectively compared against zero.
The method of \cite{kelmendi2023computing} is to first use lower bounds on sums of $S$-units \cite{evertse1984sums}, a deep result from algebraic number theory, to express $\mu$ in terms $\lambda_1,\ldots,\lambda_m \in \torus \cap \alg$.
The second step is to use Weyl's equidistribution theorem to write $\mu$ as an integral, which can be evaluated to arbitrary precision using numerical techniques.
We can use o-minimality and the Decomposition Method to give a fully geometric version of the first step.
\begin{proposition}
	\label{thm:edon}
	Let $T \subseteq \rel^d$ be semialgebraic, $M \in (\ralg)^{d\times d}$ with a decomposition $(C,D)$ over $(\ralg)^{d\times d}$, and $s \in (\ralg)^d$.
	Further let $\Dcal$ be the closure of $\{D^n \colon n\in\nat\}$ and $L$ be the limit shape of $\seq{C^{-n} \cdot T}$.
	Then 
	\[
	\mu = \lim_{n\to \infty} \frac{|\{0 \le k < n \colon M^k s \in  T\}|}{n} =  \lim_{n\to \infty} \frac{|\{0 \le k < n \colon D^n s \in C^{-n} \cdot T\}|}{n}
	\]
	is exactly equal to the measure of $L \cap( \Dcal \cdot s)$ in a suitable probability space over $\Dcal \cdot s$.
\end{proposition}
The aforementioned probability space is derived from the Haar measure on $\Dcal$.
We sketch the proof of \Cref{thm:edon}.
Write $X = \Dcal \cdot s$.
Because $X$ is bounded, by \Cref{thm:limshape-properties}, $(C^{-n} \cdot T) \cap X$ converges uniformly (with respect to the Hausdorff metric) to $L \cap X$ as $n \to \infty$.
Thus when measuring the frequency $\mu$, we can pass from the sequence  $\seq{C^{-n} \cdot T}$ to the limit shape $L \cap X$.
We thus need to understand the frequency with which $(D^ns)_{n\in\nat}$ hits $L \cap X$, which, by Weyl's equidistribution theorem, can expressed as an integral over $X$.
To check whether $\mu > 0$, we simply check, using tools from semialgebraic geometry, whether $L \cap X$ has full dimension in~$X$.

It is now understood that o-minimality gives us strong ergodic properties of linear dynamical systems, despite the fact that they are typically not compact.
This idea is pursued further in the recent work \cite{nahs}, where it shown how to compute an integral representation for the \emph{mean payoff} $\mu = \frac 1 n \sum_{k=0}^{n-1} w(M^n)$, where $M \in (\ralg)^{d\times d}$ and $w \in \rel^{d\times d} \to \rel$ is a \emph{weight function} definable in $\rexp$.
Let $(C,D)$ be a decomposition of $M$.
The idea is to write
\[
\mu = \frac 1 n \sum_{k=0}^{n-1} f(C^nD^n) = \frac 1 n \sum_{k=0}^{n-1} f_n(D^n)
\]
where $f_n \colon \Dcal \to \rel$ with $f_n (X) = C^n \cdot X$.
Then the sequence of functions $(f_n)_{n\in\nat}$, when viewed as a sequence of subsets of $\Dcal \times \rel$, converges pointwise to a limit function $f \colon \Dcal \to \rel$ (\Cref{thm:omin-effective-2}).
The next step is to argue that, in fact, $\mu = \lim_{n\to \infty} \frac 1 n \sum_{k=0}^{n-1} f(D^n)$, after which we can compute the integral representation of $\mu$ using the fact that $(D^n)_{n\in\nat}$ is ergodic by Weyl's equidistribution theorem.
We can also approximate $\mu$ to arbitrary precision using the aforementioned integral representation.
This, however, requires assuming Schanuel's conjecture since in the most general setting all we know about $f$ is that it is definable in~$\rexp$.


\bibliography{refs}

\end{document}

%% file: macros.tex
\newcommand{\torus}{\mathbb{T}}
\newcommand{\nat}{\mathbb{N}}
\newcommand{\intg}{\mathbb{Z}}
\newcommand{\rel}{\mathbb{R}}
\newcommand{\rat}{\mathbb{Q}}
\newcommand{\com}{\mathbb{C}}
\newcommand{\alg}{\overline{\rat}}
\newcommand{\ralg}{\rel \cap \alg}

\newcommand{\rexp}{\rel_{\exp}}

\newcommand{\diag}{\operatorname{diag}}

\newcommand{\Ccal}{\mathcal{C}}
\newcommand{\Dcal}{\mathcal{D}}

\newcommand{\Fcal}{\mathcal{F}}

\newcommand{\Ical}{\mathcal{I}}
\newcommand{\Jcal}{\mathcal{J}}

\newcommand{\Lcal}{\mathcal{L}}
\newcommand{\Mcal}{\mathcal{M}}

\newcommand{\Tcal}{\mathcal{T}}

\newcommand{\Mb}{\mathbb{M}}

\newcommand{\Sb}{\mathbb{S}}

\newcommand{\seq}[1]{(#1)_{n \in \mathbb{N}}}

\newcommand{\zerovec}{\mathbf{0}}
\newcommand{\im}{\bm{i}}